\documentclass[11pt]{article}
\usepackage{lmodern}
\usepackage{dsfont}
\usepackage{environ}
\input{definitions.sty}

\newif\ifshowcomments
\showcommentstrue

\ifshowcomments
\newcommand{\jason}[1]{\textcolor{red}{\textbf{Jason: }#1}}
\newcommand{\joe}[1]{\textcolor[HTML]{4169E1}{\textbf{Joe: }#1}}
\newcommand{\sdk}[1]{\textcolor[HTML]{FF00FF}{\textbf{Scott: }#1}}
\newcommand{\jens}[1]{\textcolor[HTML]{FF69B4}{\textbf{Jens: }#1}}
\newcommand{\tim}[1]{\textcolor[HTML]{FF4500}{\textbf{Tim: }#1}}
\else
\newcommand{\jason}[1]{}
\newcommand{\joe}[1]{}
\newcommand{\sdk}[1]{}
\newcommand{\jens}[1]{}
\newcommand{\tim}[1]{}
\fi

\newcommand{\defeq}{\triangleq}

\newcommand{\R}{{\mathbb{R}}}

\renewcommand{\Pr}{\mathop{\bf Pr\/}}
\newcommand{\E}{\mathop{\mathbb{E}\/}}  %

\newcommand{\dist}{\textnormal{dist}}
\newcommand{\conv}{{\textnormal{Conv}}}

\newcommand{\eps}{\epsilon}

\newcommand{\calC}{\mathcal{C}}
\newcommand{\calD}{\mathcal{D}}

\newcommand{\calX}{\mathcal{X}}

\newcommand{\bone}{\boldsymbol{1}}

\def\<{\langle}
\def\>{\rangle}

\def\vec{\bm}

\newcommand{\emailhref}[1]{\href{mailto:#1}{\tt #1}}

\begin{document}

\title{Incentive-Compatible Recovery from Manipulated Signals,\\ with Applications to Decentralized Physical Infrastructure\thanks{
Milionis's research is supported in part by NSF awards CNS-2212745, CCF-2332922, CCF-2212233, DMS-2134059, and CCF-1763970, by an Onassis Foundation Scholarship, and an A.G. Leventis educational grant.
Kominers gratefully acknowledges support from the Digital Data Design $\text{D}^3$ Institute at Harvard and the Ng Fund and the Mathematics in Economics Research Fund of the Harvard Center of Mathematical Sciences and Applications.
Roughgarden's research at Columbia University is supported in part by NSF awards CCF-2006737 and CNS-2212745.
Bonneau, Kominers, and Roughgarden hold positions at a16z crypto (for general a16z disclosures, see https://www.a16z.com/disclosures/). Milionis and Ernstberger performed work in part during an internship at a16z crypto. Notwithstanding, the ideas and opinions expressed herein are those of the authors, rather than of a16z or its affiliates. Kominers and Roughgarden also advise companies on marketplace and incentive design. 
}}

\author{
 	    \textbf{Jason Milionis}\thanks{Columbia University. Email: \emailhref{jm@cs.columbia.edu}}
        \and
        \textbf{Jens Ernstberger}\thanks{Technical University of Munich. Email: \emailhref{jens.ernstberger@gmail.com}}
        \and
        \textbf{Joseph Bonneau}\thanks{New York University, and a16z crypto. Email: \emailhref{jb6395@cs.nyu.edu}}
        \and
        \textbf{Scott Duke Kominers}\thanks{Harvard Business School, Harvard University, and a16z crypto. Email: \emailhref{kominers@fas.harvard.edu}}
        \and
        \textbf{Tim Roughgarden}\thanks{Columbia University, and a16z crypto. Email: \emailhref{tim.roughgarden@gmail.com}}
}
\date{Initial version: February 3, 2025 \\
      Current version: March 4, 2025
}

\maketitle

\thispagestyle{empty}
\vspace{-1.9em}
\begin{abstract}
We introduce the first formal model capturing the elicitation of unverifiable information from a party (the ``source") with implicit signals derived by other players (the ``observers"). Our model is motivated in part by applications in decentralized physical infrastructure networks (a.k.a.\ ``DePIN"), an emerging application domain in which physical services (e.g., sensor information, bandwidth, or energy) are provided at least in part by untrusted and self-interested parties. A key challenge in these \textit{signal network} applications is verifying the level of service that was actually provided by network participants.

We first establish a condition called source identifiability, which we show is necessary for the existence of a mechanism for which truthful signal reporting is a strict equilibrium. For a converse, we build on techniques from peer prediction to show that in every signal network that satisfies the source identifiability condition, there is in fact a strictly truthful mechanism, where truthful signal reporting gives strictly higher total expected payoff than any less informative equilibrium. We furthermore show that this truthful equilibrium is in fact the unique equilibrium of the mechanism if there is positive probability that any one observer is unconditionally honest (as would happen, for example, if an observer were run by the network owner). Also, by extending our condition to coalitions, we show that there are generally no collusion-resistant mechanisms in the settings that we consider.

We apply our framework and results to two DePIN applications: proving \textit{location}, and proving \textit{bandwidth}. In the location-proving setting 
observers learn (potentially enlarged) Euclidean distances to the source. Here, our condition has an appealing geometric interpretation, implying that the source's location can be truthfully elicited if and only if it is guaranteed to lie inside the convex hull of the observers. In the bandwidth-proving setting, we consider observers that receive noisy (and possibly throttled) evaluations of a source's bandwidth; we show that our mechanism gives a quasi-strict truthful equilibrium, meaning that the source is disincentivized from reporting a larger bandwidth than they have available.
\end{abstract}

\section{Introduction}
\label{sec:intro}

\subsection{Sources, Observers, and Manipulated Signals}

We consider a mechanism designer interested in eliciting information $x$, drawn from some abstract set $\calX$, known to a self-interested agent that we call the {\em source}. We assume that the designer cannot directly verify the accuracy of a self-report~$\hat{x} \in \calX$ by the source, but can instead rely also on the reports $\hat{\vec{y}} = \hat{y}_1,\ldots,\hat{y}_n$ of $n\ge 2$ self-interested {\em observers} that receive signals $\vec{y}$ related to $x$. We allow the source to manipulate the distribution from which observers' signals are drawn.

For example, $x$ could represent the true location of an object of interest and $\hat{x}$ the alleged location of that object (as reported by its owner, for example). Each observer~$i$ could represent a sensor, with~$y_i$ being that sensor's estimate of its distance from the object, as measured e.g.\ by the empirical round-trip time of communicating with it. The object may be able to manipulate observers' distance estimates, for example by deliberately delaying before responding to communication requests.

The primary goal of the paper is to characterize when this mechanism problem---the incentive-compatible recovery of the source's information from the (possibly manipulated and/or misreported) signals received by the mechanism---is solvable. More precisely, we ask:
\begin{enumerate}

\item Under what condition(s) on the allowable source manipulations does there exist a prior-free mechanism for which truthful behavior is a strict Bayesian Nash equilibrium?

\item Under what conditions can the truthful equilibrium be made unique?
\end{enumerate}

And conversely:
\begin{enumerate}\setcounter{enumi}{2}
\item Under what conditions is such a mechanism impossible?
    
\end{enumerate}
Our study is motivated in part by applications in decentralized physical infrastructure networks (a.k.a.\ ``DePIN''), an emerging application domain in which physical services are provided at least in part by untrusted and self-interested parties. A key challenge in such applications is how to verify the level of service that was actually provided by participants. The location-elicitation problem outlined above is a canonical DePIN application, which arises, for example, in contexts such as verifying that a resource like server or processing capacity is geographically distributed (which is important for robustness to local shocks such as weather events), as well as for confirming that decentralized data collection entities such as weather trackers are in the right place. Another canonical DePIN application is the elicitation of a source's available bandwidth, based on noisy measurements taken by observers that  may have been manipulated by the source artificially throttling its bandwidth.

We stress, however, that the model introduced in this paper is general and is not overly tailored to DePIN applications. For example, the following problem is isomorphic to the above bandwidth-elicitation problem: elicit the true ``quality'' of a candidate (student, job applicant, etc.) from noisy measurements by observers (letters of recommendation, references, etc.) that may have been manipulated 
in certain ways by the candidate (e.g., the candidate misrepresenting their abilities to the observers).

\subsection{Our Contributions}

On the modeling and analysis side, our primary contributions are the following:
\begin{itemize}

\item We introduce a novel information elicitation problem, with the key feature that the desired information is known solely to one self-interested agent (the source) who can both misreport that information and manipulate the distribution over the correlated signals observed by other agents.

\item We provide a sharp characterization of when truthful elicitation is possible in this setting: if and only if an intuitive condition that we call {\em source identifiability} holds. Intuitively, source identifiability asserts that the source's true information could in principle be recovered from an infinite number of samples from the manipulated signal distribution. In concrete examples, source identifiability translates to usable guidelines in practice.

\item We prove that whenever source identifiability fails to hold, there is no mechanism for which truthful signal reporting is a strict equilibrium.

\item When the source identifiability condition holds, meanwhile, we build on techniques from peer prediction to design a signal elicitation mechanism for which truthful reporting is a strictly optimal equilibrium for network participants, in the sense that any less informative equilibrium has strictly lower total expected payoff than is achieved under truthful signal reporting. 

\item Our mechanism's guarantee is even stronger when at least one observer is unconditionally honest with positive probability---in that case, the truthful, value-maximizing equilibrium is unique.

\item  We extend our characterization through source identifiability to coalitions, and as a consequence show that there are generally no collusion-resistant mechanisms in the settings that we consider. 

\end{itemize}

On the applied side,
our work is---to our knowledge---the first to take DePIN signal elicitation seriously as an incentive design problem. Existing DePIN frameworks have effectively ignored incentive issues by either simply assuming truthful reporting, or  
through out-of-mechanism procedures for resolving reporting issues through governance or audits. %
Our model and results offer a number of insights into DePIN applications:
\begin{itemize}

\item
We use our general results to characterize when truthful signal elicitation is possible in location signal networks and bandwidth signal networks. These two DePIN categories are actively used in practice \parencite[see, e.g.,][]{backhaul,bftpoloc}, and our results imply crucial design considerations for setting them up, as well as how signal elicitation should be conducted once these networks are deployed.

    \item In the location-proving setting, %
observers learn (potentially enlarged) Euclidean distances to the source. Here, the source identifiability condition has an appealing geometric interpretation, implying that the source's location can be truthfully elicited if and only if it is guaranteed to lie inside the convex hull of the observers. In other words, for incentive-compatible location recovery, be sure to ``surround'' with observers the possible locations of the object of interest.

\item In the bandwidth-proving setting, we consider observers that receive noisy (and possibly throttled) evaluations of a source's bandwidth; we show that our mechanism gives a quasi-strict truthful equilibrium, meaning that the source is disincentivized from reporting a larger bandwidth than they have available.

    \item Our result on equilibrium uniqueness under a mild unconditional hosesty assumption speaks to and reinforces the importance of ``decentralization'' in DePIN: this assumption seems particularly likely to hold in a large decentralized setting because because when there are many independent agents, there is a nontrivial possibility that at least one of them is not compromised; hence, in a well organized, (sufficiently) decentralized physical infrastructure network, the mere threat of being compared against an honest agent induces coordination on a truthful revelation equilibrium.
    
\item Our impossibility result for collusion-resistant mechanisms (e.g., in settings where an agent can, through sybils, act as both a source and an observer) can be interpreted as the first formal treatment of what is known as the ``self-dealing'' problem in DePIN. Our result implies that self-dealing must be handled through out-of-mechanism means, such as restrictions on permissionless entry, further refined trust assumptions, or both.

\end{itemize}
More broadly, our work here shows that DePIN networks are some of the largest and most natural applications for peer prediction and related techniques to ever arise ``in the field.''

\subsection{Related work}
\label{subsec:litrev}

Our work relates to the active and expansive body of work on peer prediction mechanisms \parencite{prelec2004,miller2005,wit2012,zhang2014elicitability,waggonerchen,prelec2021,schoenebeck2023two,kong2019information,radfalt14,kong2020information,peerneighborhood}.
A core difference relative to the peer prediction setting is that, in our work, the source is allowed to actively manipulate the other players' observed signals before those signals are elicited.\footnote{While the possibility of the source manipulating observers' signals has not been considered in the peer prediction literature, it does seem plausible that it would be a concern in some settings in which peer prediction is used in practice. For example, in settings like that of \textcite{hussam2022targeting} where peer prediction is used to elicit the ability of microentrepreneurs from assessments by their neighbors, we might imagine that, prior to participating in the peer prediction mechanism, individuals would invest effort in convincing their neighbors that they are especially effective at innovating and/or making efficient use of capital. In this sense, our work 
suggests how to augment the traditional goals of
peer prediction mechanism design to address a practical robustness concern that is typically left outside the boundaries of that model.}
Most mechanisms for the truthful elicitation of unverifiable information are surprisingly brittle (sensitive) to a number of assumptions; restrictive assumptions have been usually placed on the information structure, population size, signal spaces, and whether the mechanism is aware of the setting's joint distribution~\parencite{zhang2014elicitability,schoenebeck2023two}.
Currently, the peer prediction mechanisms with the most minimal set of assumptions to obtain ex-ante Pareto dominance to any uninformative equilibrium and strong truthfulness respectively have been given by \textcite{schoenebeck2023two} and \textcite{prelec2021} correspondingly.
The former uses a stochastically relevant setting about signals received from individuals by the nature, and the latter requires a stronger assumption than stochastic-relevance of signals, specifically second-order stochastic relevance about how one's posterior distribution about another player's signal changes, using a third player's (truthful) signal.
Generic impossibilities in peer prediction regimes with few assumptions have been given by \textcite{waggonerchen,zhang2014elicitability}; our technique for proving impossibility in non-source-identifiable model specifications is inspired by their general ideas.

In the multiple-questions peer prediction regime, to obtain truthfulness, agents are asked to report on multiple correlated tasks \parencite{dasgupta2013crowdsourced,shnayder2016informed}.
Alternatives to this method including estimating the ground truth \parencite{autograder}, including with the help of machine learning techniques, thereby almost making the problem one where partial access to the ground truth can be granted.
Relatedly, our unconditional honesty extension bears a semblance to an observation by \textcite{leytonbrown_peergrading_groundtruth} that in costly information-gathering scenarios (such as peer grading where effort has to be exerted) comparison to the ground truth with low probability is sufficient to yield a truthful elicitation mechanism; in their work, the trusted evaluator that provides an unbiased estimator of the ground truth is known in advance.

The role of the possibility of an unconditional observer in equilibrium selection is reminiscent of the role of ``commitment types" in reputation games (see, e.g., \cite{fudenbergmaskin,JARAMILLO2010416}, as well as \cite{levin2006reputation} and the references therein), although in our setting, the commitment type disciplines behavior in a single-shot mechanism rather than in a repeated game where a reputation for commitment can be observed over time.
Likewise, the need for the signal structure to be refined enough to render different strategies probabilistically distinguishable appears in various forms throughout game theory; for example, such a condition is used in characterizing when cooperation is possible in repeated games with imperfect public monitoring \parencite{fudenberg2009folk,abreu}.

The nascent literature on Decentralized Physical Infrastructure Networks (DePIN) has studied Byzantine (i.e., arbitrary adversarial) behavior in information elicitation systems, with a focus on setting limits on the fraction of the population that can be Byzantine, and assuming that the rest are unconditionally honest, without the consideration of any incentives \parencite{bftpoloc,goat,backhaul}. Our work here crucially differs in that we study the players' rational behavior according to utility functions.
\textcite{bftpoloc,backhaul} study the respective settings of location and bandwidth capacity elicitation with this in mind.
Both \textcite{bftpoloc} and likewise \textcite{goat} substantiate the practicality of using (possibly enlarged by manipulation) distances as a relevant assumption in the setting of location verification, and treat players as non-strategic; instead, the former is based on the adversarial model and performs Byzantine-resistant triangulation, while the latter considers the servers trustworthy in their timestamping.
We formally study how incentives play out with such mechanisms, and thus achieve a great synergy with high practical relevance.

\textcite{oraclefaltings}, motivated in part by the design of decentralized oracle networks, give a non-strongly-truthful peer prediction mechanism in a setting with subjective, correlated beliefs when there are binary observations. The key novelty in the model of \textcite{oraclefaltings} is the assumption that agents face some outside incentive to misreport (which depends on the aggregate outcome), and the paper focuses on how to adapt mechanisms for peer consistency \parencite{FR17} and use suitable side payments between agents to overcome these incentives; the paper also derives stronger results under assumptions about the number of agents that are unconditionally honest.
\textcite{verifierdilemma} study the specific homogeneous partially-verifiable setting of proof verification, where the status of a common object (the "proof") can be obtained by players exerting costly effort, and implement a peer prediction mechanism to address rational verifier apathy (in a blockchain context, the "verifier's dilemma"); in our setting, the model is built on the presumption of manipulability of signals received by participants.

\section{Setting}
\label{sec:model}

In this section, we will introduce the model along with our definitions.
Unless otherwise explicitly specified (e.g., when we will be discussing robustness to coalitions), all agents are assumed rational and risk-neutral.
Among all players, there is one agent (the source) which has a distinguished role, in that we are interested to elicit her (unverifiable) private information from her interaction with the rest of the players in the game induced by the mechanism.
The rest of the players have the role of observers which interact with the source and the mechanism, as described informally in \Cref{sec:intro}.

\subsection{The basic model}

A complete description of the model follows:
\begin{enumerate}
\item Nature chooses, from a joint prior distribution, the source's signal $x\in\calX$ and $n$ (private) observers' characteristics $\{p_i\}$.\footnote{Our mechanism will be independent of this distribution (i.e., prior-free). Bayesian Nash equilibria of the mechanism are with respect to this prior. We assume that $x$ can take on at least two different values and that $n \ge 2$.}
\item The source chooses an $n$-dimensional distribution $\calD$ either from $L_x$, where $L_x$ is a feasible set of distributions of reports (according to application-specific modeling), or any other distribution that does not correspond to any feasible distribution if the source were truthful. Formally, the source chooses $\calD \in L_x \cup \left\{\hat{\calD} \mid \forall x\in\calX: \hat{\calD} \not\in L_x \right\}$. We denote by $L$ the multi-valued function defined by $f(x) \defeq L_x$ wherever the context is clear, and we term $L$ the \emph{model specification}.
\item Nature chooses $\vec{y} \sim \calD$, and each $y_i$ gets sent to every one of the $n$ observers (each one privately observes their own signal).
\end{enumerate}

The observers and source then participate in a mechanism $M$, with common knowledge of all information above, including the model specification $L$.

This model allows potentially for the source to pick among adversarial values, if the distributions belonging to each $L_x$ are modeled as point masses. In that special case, the set of distributions is then a set of points, out of which the source may choose their favorite one.

Note that in this paper we will consider discrete signal spaces. Our work can be generalized to continuous signal spaces by using techniques in a similar fashion to \textcite{schoenebeck2023two,radfalt14,peerneighborhood}.

We move on to define our condition (\Cref{def:rel}) that we will tie to the existence of a mechanism where signal-truthfulness is a strict Bayesian Nash equilibrium.
We use the standard definitions for the Bayesian Nash equilibrium in games with incomplete information.

\begin{definition}
[Source identifiability]
\label{def:rel}
A source in a model specification $L$ is called \emph{identifiable} if for any two different $x_1\ne x_2$, \[ L_{x_1} \cap L_{x_2} = \emptyset \,, \] i.e., there exists no distribution that's exactly the same for two different source signals.
Equivalently, a source is identifiable if and only if the \emph{multi-valued} function defined by $x \rightrightarrows L_x$ is injective.
\end{definition}

We call this property \emph{identifiability}, because in line with statistics, it roughly implies that the model's parameters can be uniquely determined from the probability distribution of the observed data. In other words, if one could somehow \emph{perfectly} observe the true data-generating process---e.g., with an infinite amount of data---they would be able to uniquely deduce the value of the parameter from that distribution.
Thus, our mechanism's intuition is to make use of strictly proper scoring rules to ensure that we can truthfully obtain the source's value, given the rational observers are honest in their signal reporting. We stress that \Cref{def:rel} allows for distributions in two  different sets $L_{x_1}$ and $L_{x_2}$ that are arbitrarily close to each other (e.g., in total variation distance), and only forbids identical distributions.

\subsection{Example: proof of location}
\label{subsec:setting_loc}

One example we referred to in \Cref{sec:intro} was \textit{location verification}.
We can now see how this maps to the formalism in our model, in the following way: Suppose that both the source and observers are located somewhere on the plane.
The observers' locations on the plane are known and in our model, correspond to vectors $p_i\in\R^2$.
The mechanism designer's objective is to estimate the (a priori unknown) location of the source, which is going to be $x\in\R^2$.
Observers gather information from the source, which consist of positive numbers $y_i$ that are interpreted as the distances between observer $i$ and the source.

For this example, we suppose that the source can misrepresent its distance to each observer, but can only artificially \emph{increase} its distance to each one individually (e.g., by delaying communications); it cannot make its distance seem smaller than it actually is.
In this sense, this example allows arbitrary ``one-sided manipulation'' by the source.
This constraint would be represented with our model specification as $L_x$ (the feasible set of reports) being a (possibly uncountable) set of point-mass distributions: the set of all potentially enlarged distances to each observer.
The source is therefore able to choose its favorite enlarged distances that each observer individually receives.\footnote{The randomization by nature of $\vec{y}\sim\calD$ is meaningless in this example, as every ``distribution" is just a point mass.}

What does \Cref{def:rel} translate to in this setting? In \Cref{subsec:main_loc}, we show that source identifiability translates to a convex hull condition: a source's location is identifiable if and only if all possible locations of the source are contained in the convex hull formed by the observers' locations.\footnote{For the exact formalism and details, we refer the interested reader to \Cref{subsec:main_loc}.} This convex hull condition is intuitive and---importantly---gives guidance for how observers should be positioned in practice.

\section{Main results}
\label{sec:main_results}

We begin with our impossibility result for a signal-truthful mechanism in the case of a model specification where the source is not identifiable.

\begin{theorem}
[Impossibility when source is not identifiable]
\label{thm:impossibility}
Given any model specification $L$ where the source is not identifiable, i.e., does not satisfy \Cref{def:rel}, there exists no mechanism $M$ taking as input not only the players' self-declared $\hat{x}, \hat{\calD}, \hat{\vec{y}}$ but also the model specification $L$, for which signal truthfulness is a strict Bayesian Nash equilibrium.
\end{theorem}
\begin{proof}
For the sake of contradiction, assume there was such a mechanism $M$, and that it assigns a payoff $u((\hat{x}, \hat{\calD}), \hat{\vec{y}}, L)$ to the source.
Because the source in $L$ is not identifiable, there exist $x_1\ne x_2$ and a joint distribution of manipulated observer signals $\calD$ such that $\calD \in L_{x_1} \cap L_{x_2}$.

Since signal truthfulness is a strict Bayesian Nash equilibrium for $M$, call the respective strategy profile functions $(s_0(\cdot), s_1(\cdot),\dots,s_n(\cdot))$, where $s_0$ denotes the source's strategy, %
mapping the private values of each player to their actions in the mechanism (the actions are declaring $(\hat{x}, \hat{\calD})$ for the source and $\hat{y}_i$ for observer $i$); then, it must be that
\begin{align}
\E_{\vec{y}\sim\calD} \left[ u((x_1, \calD), \vec{y}, L) | x_1 \right] > \E_{\vec{y}\sim\calD} \left[ u((x_2, \calD), \vec{y}, L) | x_1 \right]
\label{eq:imp1}
\\
\E_{\vec{y}\sim\calD} \left[ u((x_2, \calD), \vec{y}, L) | x_2 \right] > \E_{\vec{y}\sim\calD} \left[ u((x_1, \calD), \vec{y}, L) | x_2 \right]
\label{eq:imp2}
\end{align}

Build the following rogue (i.e., non-truthful) strategy where the source is truthful when its private value is $x_1$ but behaves the same for $x_2$ (obviously the truthful $\calD$, chosen by the source, is feasible for both signals $x_1, x_2$ by the model specification), i.e., $s_0'(x_2) = (x_1, \calD)$ and otherwise $s_0'$ is the same as $s_0$. We now prove that, since this gives the same expected payoff to the source (conditioning on $x_1$) as the truthful strategy, the Bayesian Nash equilibrium cannot be strict, which is the contradiction finishing the proof.

Indeed, we have that
\begin{align*}
&\E_{\vec{y}\sim\calD} \left[ u((x_2, \calD), \vec{y}, L) | x_2 \right]
=
\E_{\vec{y}\sim\calD} \left[ u((x_2, \calD), \vec{y}, L) | x_1 \right]
<
\E_{\vec{y}\sim\calD} \left[ u((x_1, \calD), \vec{y}, L) | x_1 \right]
=
\\
&\E_{\vec{y}\sim\calD} \left[ u((x_1, \calD), \vec{y}, L) | x_2 \right]
<
\E_{\vec{y}\sim\calD} \left[ u((x_2, \calD), \vec{y}, L) | x_2 \right]
\,,
\end{align*}
which is a contradiction, and where the equalities hold because the conditional distribution $\calD$ is the same and the conditioned random variable is independent of the conditioning random variable, and the inequalities are \Cref{eq:imp1,eq:imp2} respectively.
\end{proof}

We move on to the positive results, and give a mechanism to truthfully elicit the unverifiable information of the source and observers, subject to \Cref{def:rel}.
For technical convenience, and without loss of generality, we will also make the following assumption which is roughly stochastic relevance \emph{conditioned on the source's truthfulness}:\footnote{Because the elicitation of the source's signal is the final sought-after consequence, our results can be generalized to the case that the technical condition does not hold, and the optimal strategy is a quasi-strict equilibrium where observer $i$ submits any $y_i$ that---conditioned on the truthful $x$---yields the exact same marginal distribution for the rest of all observers, i.e., the strategy groups the non-distinct (in terms of the joint probability distribution) $y_i$'s.}
\begin{assumption}
[Technical Condition]
For any $x\in\calX$, distribution $\calD\in L_x$, $i\in [n]$, and two $y_i \ne y_i'$,
\[
\Pr_{\calD|y_i}[\vec{y}_{-i} | y_i]
\ne
\Pr_{\calD|y_i'}[\vec{y}_{-i} | y_i']
\,,
\]
i.e., there do not exist two different $y_i \ne y_i'$ that induce the same conditional distribution (for the truthful $x$) on the rest of all truthfully-received observers' signals.
\label{ass:cond}
\end{assumption}
\Cref{ass:cond} effectively means that $y_i$ causes the posterior of any observer $i$ to change based on the (truthful) value they receive from the source. In most common regimes, such an assumption would hold, for example because the source has non-overlapping sets of $y_i$'s (c.f., \Cref{subsec:main_loc}), or because each of the observers obtains an independent estimate centered on the source's quality of service (c.f., \Cref{subsec:bandwidth}).

The sub-mechanism that we will use to gather information from the observers about the source belongs to the class of Bayesian Truth Serum (BTS) mechanisms, pioneered by \textcite{prelec2004}; we specifically use one of the mechanisms in \textcite{prelec2021}, although we remark that similar theorems to ours could be proven using many other similar mechanisms, developed by \textcite{prelec2021,schoenebeck2023two}.

The mechanism $M$ is presented in \Cref{m}.
We denote by $\bone\{A\}$ the indicator function that is 1 if $A$ happens, otherwise 0.
Recall that a strictly proper scoring rule is a (potentially extended) real-valued function $P(\calD, \vec{y})$ that takes as input a probability measure $\calD$ and a realized outcome $\vec{y}$, and outputs a real number (reward) such that
\[
\E_{\vec{y} \sim \calD}[P(\calD', \vec{y})]
\le
\E_{\vec{y} \sim \calD}[P(\calD, \vec{y})]
\text{ for all distributions } \calD, \calD'
\,,
\]
with equality if and only if $\calD'=\calD$.
\begin{algorithm}
\caption{Mechanism $M$ run after model with inputs $(x, \calD), \vec{y}\sim\calD$}\label{m}
\begin{enumerate}[label=\arabic*.]
\item Observers submit $\hat{y}_i$ to the mechanism.
\item The source submits $(\hat{x}, \hat{\calD})$, where $\hat{x}\in\calX$, to the mechanism and to the observers.
\item Observers submit $\pi_i \in (0,1]$ and $\hat{x}_i \in \calX \cup \{\emptyset\}$ to the mechanism.
\item Each observer $i$ is paired (by the mechanism) with a random observer $j$, and submits a probability distribution for $j$'s signal to the mechanism.\footnotemark\ The probability distribution is defined by non-negative numbers $\hat{q}_i(\cdot)$ that sum to 1 across $j$'s support of signals.
\item Each observer $i$ obtains reward
\[
  \log\left(\frac
  {\hat{q}_i(\hat{y}_j)}
  {\pi_j}\right)
- \left| \log\left(\frac
  {\hat{q}_j(\hat{y}_i) \pi_j}
  {\hat{q}_i(\hat{y}_j) \pi_i}\right)\right|
+
\bone\{\hat{x}_1 = \dots = \hat{x}_n\}
\,.
\]
\item The source obtains reward
\begin{align}
P(\hat{\calD}, \vec{\hat{y}})
+
\bone\{\hat{x} = \hat{x}_1 = \dots = \hat{x}_n\}
\,,
\label{eq:prov_util}
\end{align}
where $P(\cdot, \cdot)$ is any strictly proper scoring rule.
\end{enumerate}
\end{algorithm}
\footnotetext{We note that, per standard procedure in peer prediction mechanisms~\parencite[see, e.g.,][]{schoenebeck2023two}, one need not ask for an entire probability distribution, but just a single probability (at least in the discrete signals case) by the mechanism choosing a random value as a virtual signal and asking $i$ for the probability that $j$'s signal is that virtual signal; $i$'s reward is then to be modified such that if the randomly chosen signal value matches the actually submitted value from $j$ then the normal reward function is followed, otherwise a (maximal) reward of $0$ is given.}

This mechanism is prior-free. Further, the mechanism does not require that $\hat{\calD}\in L_{\hat{x}}$, and for this reason is also free of the model specification $L$.
In other words, the mechanism need not know the model specification at all, and our analysis of the mechanism holds so long as the true (private) signals of the observers indeed come from that model.

We next prove a number of desirable properties of this generic mechanism.
To state it, we first define the signal-truthful strategy profiles:

\begin{definition}
\label{def:sigtruth}
We call a strategy profile \emph{signal-truthful} if:
\begin{itemize}
    \item The source, given $x$, chooses $\calD\in L_x$,\footnote{Note that our theorem will state that there exists some $\calD$ for a signal-truthful strategy profile; not all $\calD$'s might correspond to signal-truthful profiles that are strict Bayesian equilibria.} and then submits $(\hat{x}, \hat{\calD}) = (x, \calD)$.
    \item Observer $i$, given $y_i$, the source's strategy $(\hat{x}, \hat{\calD})$, and pairing $j$, submits $\hat{y}_i = y_i, \pi_i = c \cdot \hat{\calD}_i(y_i)$ (the probability of the marginal on $i$ to get $y_i$ as per $\hat{\calD}$, rescaled by any $0 < c \le 1$ which is fixed across observers), $\hat{x}_i$ such that $\hat{\calD} \in L_{\hat{x}_i}$ (unique by source identifiability) or $\hat{x}_i = \emptyset$ if none exists, and $\hat{q}_i(\cdot)$ to be the posterior on $j$'s signal conditional on $y_i$ as per $\hat{\calD}$.
\end{itemize}
\end{definition}

\begin{theorem}
[Truthful equilibrium]
\label{thm:mechanism}
For any $x\in\calX$ (i.e., any prior on $\calX$) and any model specification $L$ where the source is identifiable as per \Cref{def:rel}, and subject to \Cref{ass:cond}, there exists $\calD\in L_x$ such that for any $0 < c \le 1$, the signal-truthful strategy profiles as defined in \Cref{def:sigtruth} with the choice of $\calD$ are strict Bayesian Nash equilibria of the game induced by the model and the mechanism $M$, where strictness is defined disregarding (i.e., aggregating over) any distribution $\calD\in L_x$ for the truthful $x$, for all $x\in\calX$.\footnote{Recall that this does not detract from signal truthfulness by source identifiability.}
Additionally, for any less informative equilibrium of the mechanism, there exists a signal-truthful equilibrium with strictly higher total expected payoff.
\end{theorem}
\begin{proof}
First, we prove strict truthfulness. Consider the source and observers separately.\footnote{In what follows, because of the aggregating notion of strictness explained in the theorem's statement, we show that, in the extensive form game, strictness is satisfied disregarding (i.e., conditioning on) the choice $\calD\in L_x$ that the source makes in the first step of the game before mechanism $M$.}
\begin{itemize}
    \item For the source, assuming all observers are truthful ($\hat{\vec{y}} = \vec{y}$): the source selects some $\calD \in L_x \cup \left\{\hat{\calD} \mid \forall x\in\calX: \hat{\calD} \not\in L_x \right\}$. By the strict properness of scoring rule $P$, the unique best-response is to submit $\hat{\calD}=\calD$. Because $\hat{\calD}=\calD$, observers will choose $\hat{x}_i = x$ for all $i$ by \Cref{def:rel} if the chosen $\calD\in L_x$, otherwise they will choose $\emptyset$ (which is infeasible for the source to report, as it's the special signal of the observers that the source was not truthful), since (again by source identifiability) there is no other $x' \ne x$ that has the same distribution $\calD$. Strictness for $\hat{x} = x$ follows.
    \item For observer $i$, assuming all other observers and the source are honest (in particular, this means $\hat{\calD} = \calD \in L_x$): first, $\hat{x}_i = \hat{x}$ is the unique best-response by \Cref{def:rel}. Second, by the stochastic relevance of \Cref{ass:cond} conditioned on the source's truthfulness hence a distribution $\calD \in L_x$, the submechanism among the observers operates as a strictly truthful peer prediction mechanism \parencite{prelec2021}. Strict truthfulness for the rest of the strategic choices of observer $i$ follows by the basic mechanism's strict truthfulness.
\end{itemize}

It is left to prove the second part of the theorem.
Any less informative equilibrium in $M$ exhibits either pooling on $\hat{x} = \hat{x}_i \ne x$ or is a less informative equilibrium of the sub-mechanism with observers.
In the latter case, first consider the associated payoffs of the observers based only on their reports except for $\hat{x}_i$'s. Applying the data processing inequality twice \parencite[see, e.g.,][]{prelec2021}, any signal garbling equilibrium that is less informative (by either randomizing or pooling over a strategy) has strictly less expected payoffs for every observer than the corresponding signal-truthful equilibrium. Therefore, there exists a $c < 1$ such that the total expected payoff (including the source) of the corresponding signal-truthful equilibrium is strictly higher.
For the former case, we repeat the latter argument, because the equilibrium with $\hat{x}=\hat{x}_i=x$ is a tie in the individual expected payoffs conditional on each player's signals.
This proves the second part of the theorem.
\end{proof}

A common observation with some peer prediction mechanisms \parencite{schoenebeck2023two} is that it is sometimes hard to imagine that players would arrive at non-truthful equilibria that require unnatural coordination in their play.
In mechanism $M$ and any signal-truthful equilibrium, $\hat{\calD}=\calD$ provides a natural point for reports on $\hat{x}, \hat{x}_i$ to pool on; any other choice of $\hat{\calD}$ by the source would in expectation provide them with strictly less payoff, therefore the aforementioned equilibrium could be a natural coordinating strategy profile in the extensive-form game.

\section{Extensions}
\label{sec:ext}

\subsection{Unconditional honesty}
\label{subsec:uncond}

The guarantee of \Cref{thm:mechanism} can be sharpened further whenever there is positive probability that at least one observer is unconditionally honest: 

\begin{lemma}
[Any probability of observer unconditional honesty yields unique truthful equilibrium]
\label{lem:honest}
If there is a positive probability that any one observer is unconditionally honest, then the truthful equilibrium is the only equilibrium of mechanism $M$.
\end{lemma}

Unconditional honesty of any (random) observer in the game turns the extensive-form game into one where any (other) observer's information set cannot feasibly have an implicit guarantee around their pair's behavior given by a Bayesian Nash equilibrium; this is why they must randomize over the (non-trivial) possibility that they get paired with the unconditionally honest observer. It turns out that the mere threat of being matched up to such an observer is enough to deter non-truthful, less-informative equilibria from forming in the game. This---along with application of the implications of \Cref{thm:mechanism}---is the reason why the only feasible (unique) equilibrium is the signal-truthful one. 

\begin{proof}[Proof of \Cref{lem:honest}]
Name the probability $p_0>0$, and say observer $j$ is unconditionally honest with probability $p_0$. Then, by the strictness of the truthful equilibrium in $M$, if observer $i \ne j$ played any strategy other than the truthful one, then with probability $p_0/n$ they would obtain strictly less than the maximal payoff achieved with the truthful strategy (because they got paired with a truthful observer), and with probability $1-p_0/n$ they will obtain a payoff that (by the second part of \Cref{thm:mechanism}) is in expectation less than the truthful one. Thus, by $p_0/n>0$, an equilibrium is only possible if all observers report truthfully, therefore by strictness of the truthful equilibrium of $M$, the source will also be truthful.
The lemma follows.
\end{proof}

The intuition and formal argument for \Cref{lem:honest} make it clear that the role that (the possibility of) unconditional honesty plays here reflects a general idea, which we suspect may be useful more broadly: In peer prediction--based mechanisms, agents' reports are cross-examined against each other---and the possibility that at least some agents may be unconditionally honest means that any putative non-truthful equilibrium behavior has some risk of being identified, and punished, through cross-examination with an unconditionally honest agent (who always reports truthfully). Thus, even a small positive probability of an unconditionally honest agent helps isolate the truthful equilibrium.

We note also that the assumption that at least one observer might be unconditionally honest  is particularly natural in the context of large signal networks with many independent participants---like in the DePIN applications we examine in \Cref{sec:appl}. Indeed, with many independent observers, it becomes increasingly reasonable to assume that the each observer believes that at least one observer may not be compromised. (Moreover, in many applications, it may be possible for the signal network's organizer to directly guarantee that at least one observer is unconditionally truthful---perhaps by managing that observer themself in a way that is common knowledge.)

\subsection{Individual rationality}

In mechanism $M$, with respect to an arbitrary source, the observers can guarantee non-negative expected payoffs if they behave according to a signal-truthful equilibrium as per \Cref{thm:mechanism}.
More specifically, according to straightforward calculations from the mechanism's payoffs, we note that the expected payoff of $M$ to any observer (conditional on their signal) if all observers behave truthfully is non-negative, because it is exactly the Kullback-Leibler divergence between the posterior probability distribution of $j$ conditional on $i$'s signal and the marginal distribution of $j$'s signal according to $\calD$. This divergence is guaranteed to be non-negative.

For the source, the usual comments applicable to affine transformations of scoring rules to guarantee individual rationality hold: for example, if we choose the quadratic scoring rule, then indeed, by adding $1/2$ for a transformed scoring rule, the payoff to the source is always non-negative.

\subsection{Collusion of source with observers}
\label{subsec:collusion}

A significant concern in decentralized systems is collusion. A commonly cited reason is that collusion can be readily facilitated with smart contracts that provide the mechanism for parties to coordinate and credibly commit to prescribed behavior. In this section, we will be particularly concerned with collusion of (a subset of) observers with the source, and show that it is essentially impossible for a mechanism to be collusion-resistant and strictly truthful.

\begin{definition}
[Source-observers collusion-resistance]
Consider a (specific) subset of observers $\calC\subseteq[n]$ that collude with the source.
In our setting, we will call a mechanism \emph{$\calC$-collusion-resistant}, if and only if for any joint (coordinated) reports of the source and subset $\calC$, strict truthfulness holds for the source's value, i.e., it is a strict best-response for the source to report its true value to the mechanism. 
\end{definition}

We note that this definition is akin to a quasi-strictness definition, because it aggregates over the actions of the other colluding players in the game induced by the mechanism and the model.

\begin{lemma}
\label{lem:collusion}
Assume that $\calC\subseteq[n]$ is common knowledge to all players and the mechanism. For any model specification $L$, consider the following refining as a multi-valued function $x \rightrightarrows L_x|_\calC \defeq \{ \calD_{\overline{\calC}} \mid \calD \in L_x \}$, i.e., every distribution is a marginal of the original model specification over all observers not in the colluding set $\calC$.
Unless source identifiability holds for the model specification defined by $L|_\calC$, there is no mechanism that can be $\calC$-collusion-resistant where signal-truthfulness is a Bayesian Nash equilibrium (in the same sense as in \Cref{thm:mechanism}).
\end{lemma}
\begin{proof}
Forward: Construct an instantiation of mechanism $M$ (\Cref{m}), where the mechanism only operates over the subset $\overline{\calC}$ of all observers that do not collude; the mechanism otherwise ignores (does not request) input from observers in $\calC$. By source identifiability on $L|_\calC$ and \Cref{thm:mechanism}, the desired properties hold.

Reverse: In the framework of our impossibility proof in \Cref{thm:impossibility}, by the coordination of source and observers' actions, effectively the actions/reports of players in $\calC$ are dictated by the source. Therefore, the source's expected payoff ranges only over $\vec{y}_{\overline{\calC}} \sim \calD_{\overline{\calC}}$ for some $\calD_{\overline{\calC}}\in L_x|_\calC$.
Now, as in \Cref{thm:impossibility}, by the source's non-identifiability in $L|_\calC$, consider two different $x_1 \ne x_2$ for which the conditional distribution can be chosen by the source to be the same, i.e., it holds that $\calD_{\overline{\calC}}\in L_{x_1}|_\calC \cap L_{x_2}|_\calC \ne \emptyset$.
Following the rest of the expected payoffs gives rise to a similar contradiction like in \Cref{thm:impossibility}.
\end{proof}

In the context of decentralized physical infrastructure networks (DePIN), whereby participation to a mechanism on the blockchain is generally permissionless and unconstrained, i.e., new players are free to join the mechanism, one special case of such collusion of a source with a subset of observers is when these "observers" are the source itself.
This is referred to as \emph{self-dealing} in the context of DePIN, and our \Cref{lem:collusion} above essentially proves that it is impossible to handle, at least in a prior-free mechanism.
Thus, we formally prove that self-dealing \emph{must} be handled out of mechanism, either via restrictions to permissionless entry, further refined trust assumptions, or both.

\section{Applications}\label{sec:appl}

\subsection{Location signal networks}
\label{subsec:main_loc}

Continuing the discussion of location verification we began in \Cref{subsec:setting_loc}, we have that the mechanism designer wants to estimate the source's location, and use the observers' information gathering to properly incentivize them to conclude the actual source's location.

More specifically, $x\in\R^d$ is a vector of a Euclidean space, and each observer's location is fixed as $p_i\in\R^d$. The model specification consists of (possibly enlarged) distances to the source so long as these are plausibly feasible by some other $x'\in\calX$ in the model, and is given in \Cref{ass:euclid}.\footnote{Everybody knows that the source is somewhere on $\calX$ by the common knowledge property. It should not be possible for the source to enlarge their distances such that they claim some $x'\ne x$ in order for it to be identifiable, but we include it in the fully general model specification.}
\begin{env}[Location signal network]
\label{ass:euclid}
The source's location is a point  $x\in\calX\subset\R^d$. Observers are represented by points $p_i\in\R^d$, and $L_x = \{ \{(y_1, y_2, \dots, y_n)\} \mid y_i \ge \dist(p_i, x)\ \forall i \text{ and } \exists x'\in\calX : \forall i : y_i = \dist(p_i, x')\}$, i.e., every distribution that belongs to $L_x$ is just a point mass, and feasible reports of the source include all individual values greater than its (minimum) distance to observer $p_i$ that are consistent with some feasible $x'\in\calX$.
\end{env}

As a matter of fact, the definition of this model specification means that the source can claim \emph{any} potentially enlarged distances to observers, not just the plausibly feasible ones. This is because, according to the model description in \Cref{sec:model}, the full set of potential source choices to be revealed to observers, i.e., $L_x \cup \left\{\hat{\calD} \mid \forall x\in\calX: \hat{\calD} \not\in L_x \right\}$, includes the full set of strategic choices $\{ \{(y_1, y_2, \dots, y_n)\} \mid y_i \ge \dist(p_i, x)\ \forall i\}$.\footnote{It includes many other possible lies of the source as well, but the particular ones of enlarged distances are of interest, as described in \Cref{sec:intro,subsec:setting_loc}.}
Therefore, the source may report any individual values that are larger than the actual distance; of course, by the guarantees of \Cref{thm:mechanism}, they will only be strictly worse off if they do choose to do so and observers are signal-truthful, if the source is identifiable according to \Cref{def:rel}.

We remark that common alternative noisy models also fall into our framework, e.g.,
\[
L_x = \left\{ 
\left\{
\begin{aligned}
  &(y_1 + \eps_1, \dots, y_n + \eps_n) \text{ w.p. } 1/2, \\
  &(y_1 - \eps_1, \dots, y_n - \eps_n) \text{ w.p. } 1/2
\end{aligned}
\right\}
\;\middle|\; y_i \ge \dist(p_i, x)\ \forall i
\text{ and }
\exists x'\in\calX : \forall i : y_i = \dist(p_i, x')
\right\}
\,.
\]
Similar noisy models can represent observers which ping the source and are well-suited to participate in our mechanism, since they can readily provide posterior distributions by virtue of noise estimates from such links they have with the source.

\Cref{prop:loc} gives a sufficient and roughly necessary condition that characterizes the truthful elicitation of the source's location: a mechanism with a strictly truthful Bayesian Nash equilibrium can be given if and only if the source is guaranteed to lie inside the convex hull of the observers.
Note that for the necessity, we have to exclude trivially distinguishable cases, such as $\calX$ being just two points outside the convex hull on opposite sides of it.
To overcome these, since such trivial cases do not add value to the characterization, we require (to prove that the source is not identifiable in these cases) that a non-measure zero (in $\R^d$) mass outside of the convex hull is included in $\calX$.

In practice, the condition of \Cref{prop:loc} is very actionable: it indicates that one should think about where the source $x$ might be (in $\R^d$), and make sure to "surround" it on the perimeter with sensors.

We note that in this setting, the source's reward attains a particularly satisfying format: any scoring rule $P(\hat{\calD}, \hat{\vec{y}})$ rewards consistency at the signal-truthful equilibrium; either the vectors obtained by the observers (which according to \Cref{prop:loc} cannot be manipulated) match exactly the claimed ones by the source (which may be arbitrary, since they don't need to conform to any guidelines according to the model specification) in which case this component of the source's reward is maximized, or the source does not obtain the maximum reward.

In what follows, we denote by $\conv(\{p_1,\dots,p_n\})$ the convex hull defined by the points $\{p_1,\dots,p_n\}$.
We move on with two helpful lemmas about Euclidean spaces, whose proofs we include in \Cref{app:loc} (\Cref{lem:loc_first} concerns the injectivity of exact distances on any domain that is a subset of the convex hull, and \Cref{lem:loc_sec} is about their distance vectors being coordinate-wise incomparable) that will be used to prove \Cref{prop:loc}.

\begin{lemma}
\label{lem:loc_first}
The map $x\mapsto (\dist(p_1,x),\dots,\dist(p_n,x))$ is injective in any domain $\calX$ that is a subset of the convex hull $\conv(\{p_1,\dots,p_n\})$.
\end{lemma}

\begin{lemma}
\label{lem:loc_sec}
Consider two $x', x \in \conv(\{p_1,\dots,p_n\})$. If it holds that $\dist(p_i, x') \ge \dist(p_i, x) \;\forall i$, then $x'=x$. (The converse is trivial, since all distances are the same.)
\end{lemma}

\begin{proposition}
[Convex hull characterization]
\label{prop:loc}
In the model defined by \Cref{ass:euclid}, if $\calX\subseteq \conv(\{p_1,\dots,p_n\})$, then the source is identifiable.
Conversely, if $\calX$ is a superset of a non-measure zero mass of points outside $\conv(\{p_1,\dots,p_n\})$, then the source is not identifiable.
\end{proposition}
\begin{proof}
Forward: Recall that we need to show that the multi-valued function $x \rightrightarrows L_x$ is injective.
As a result of \Cref{lem:loc_sec}, any enlarged distances fall outside of the (truthful) model specification $L_x$, because they are not plausibly feasible by any other truthful $x'\in\calX\subseteq \conv(\{p_1,\dots,p_n\})$.
Therefore, $x \rightrightarrows L_x$ corresponds exactly to the map that \Cref{lem:loc_first} proves is injective, and this direction is complete.

Reverse: If $\calX$ is a superset of a non-measure zero set of points outside of $\conv(\{p_1,\dots,p_n\})$, then there are two different $x_1 \ne x_2\in\R^d$ and a separating hyperplane from the convex hull (represented by its unit normal vector $u\in\R^d$) such that $\forall i: \<p_i, u\> \ge 0$ and $x_1 = -\alpha u, x_2 = -\beta u$ for some $\alpha, \beta > 0$.
Without loss of generality, order $x_1, x_2$ such that $\beta > \alpha$.
We show that $\{(\dist(p_1, x_2), \dots, \dist(p_n, x_2))\} \in L_{x_1} \cap L_{x_2} \ne \emptyset$, therefore the source is not identifiable.
Indeed, it suffices to prove that $\forall i : \dist(p_i, x_2) > \dist(p_i, x_1)$.\footnote{Notice that here, the quantifier "for all $i$" is the non-trivial part, and why we use the co-linear vectors $x_1, x_2$ with the hyperplane's normal vector $u$.}
This is true by computation, since for any $i$:
\[
\| p_i - x_2 \|^2 - \| p_i - x_1 \|^2
=
(\beta - \alpha) (\beta + \alpha + 2\<p_i, u\>)
> 0
\,.
\qedhere
\]
\end{proof}

\subsection{Bandwidth signal networks}
\label{subsec:bandwidth}

In this setting, the mechanism designer wants to elicit the source's (ideally maximum available) bandwidth.
Observers obtain noisy and possibly throttled estimates about the source's bandwidth; the model specification is given in \Cref{ass:band}.
A primary rationale for this model specification is the observation that internet connections between two nodes might be throttled, and internet links can operate over multiple hops, therefore even though an observer might have the capacity to notice the full declared bandwidth of the source if connected through a direct peer-to-peer link, they may in fact be connected via a set of intermediate nodes that cannot support this bandwidth.
The model, then, would reasonably be expected to be unable to certify a high connection speed, if no observer can witness it.
Thus, the model specification below also bakes in the assumption that there is at least one observer capable of probabilistically observing the actual source's bandwidth.

\begin{env}[Bandwidth signal network]
\label{ass:band}
The source's bandwidth is $x\in\R^+$.
Given $x$, every observer obtains independent estimates of the source's bandwidth, coming from distributions whose support is upper bounded (or truncated) at some value that's at most $x$, i.e., $L_x = \{ \calD_1^x\times\dots\times\calD_n^x \mid 0 \le \mathrm{support}(\calD_i^x) \le x \ \forall i\}$, and $\exists \calD_1^x\times\dots\times\calD_n^x \in L_x$ such that $\exists i : x\in\mathrm{support}(\calD_i^x)$.\footnote{This is the condition we impose, because we remind that we consider discrete distributions. Otherwise, we need to impose non-zero measure in a continuous distribution, i.e., $\calD_i^x(x) > 0$.}
\end{env}

Unfortunately, most settings following \Cref{ass:band} are not source-identifiable, as \Cref{prop:band1} proves.
\begin{proposition}
\label{prop:band1}
In the model of \Cref{ass:band}, there is at least one model specification where the source is not identifiable.
\end{proposition}
\begin{proof}
There are many example instantiations of the generic model given by \Cref{ass:band} that do not satisfy source identifiability.

For example, consider the further refined model, where some of the included distributions (let's denote them by $\calD_i$) in the product distributions contained in $L_x$ (among others) are distributions upper bounded at some fixed value $p_i$, i.e., $\mathrm{support}(\calD_i) \le \min\{p_i, x\}$.
We can model this way the source's choice to \emph{artificially throttle} the bandwidth that it appears that it has to each of the observers; note that in most realistic regimes, this option is practically available to the source. The source  can then (strategically) choose these throttled distributions---perhaps to its detriment in a system where high bandwidth is incentivized.

Formally, for any two different $x_1, x_2$ such that $x_1 > x_2 > \max\limits_{i\in [n]}\{p_i\}$, it is clear that
\[
\{\calD_1\times\dots\times\calD_n\} \subseteq L_{x_1} \cap L_{x_2} \ne \emptyset
\,;
\]
hence, the source is not identifiable according to \Cref{def:rel}.
\end{proof}

We can now derive a modification of the given guarantees; specifically, we first relax the strictness requirement, as follows.
\begin{definition}
\label{def:band_quasistrict}
A signal-truthful strategy profile of mechanism $M$ will be called \emph{quasi-strict for the source}, if any $\hat{x} > x$ attains strictly less payoff for the source when the observers are following the specified signal-truthful strategies.
\end{definition}

The relevant lemma follows in \Cref{lem:band}.

\begin{lemma}
\label{lem:band}
For any prior on $\calX$, signal-truthfulness defined by \Cref{thm:mechanism} in mechanism $M$, where\footnote{We need to specify the strategy, because for any given $\hat{\calD}$, there is no longer a unique $\hat{x}_i$ such that $\hat{\calD} \in L_{\hat{x}_i}$, due to the source not being identifiable.} we additionally refine the strategy of every observer by reporting $\hat{x}_i = \max\limits_i\left\{\mathrm{support}(\hat{\calD}_i)\right\}$ from the received $\hat{\calD}$, in the setting defined by \Cref{ass:band}, is quasi-strict for the source, as defined by \Cref{def:band_quasistrict}.
\end{lemma}
\begin{proof}
Modifying the proof of \Cref{thm:mechanism}, for the source's strategy only, by strict properness of the scoring rule, it's still going to be that $\hat{\calD} = \calD$ for some $\calD \in L_x$ that the source chooses.
The source can attain the additional reward of $1$ from the indicator function \emph{and} with every challenger reporting $\hat{x}_i = x$ according to the signal-truthful Bayesian Nash equilibrium, by choosing $\calD$ appropriately, since by \Cref{ass:band}, $\exists \calD \defeq \calD_1^x\times\dots\times\calD_n^x \in L_x$ such that $\exists i : x\in\mathrm{support}(\calD_i^x)$.\footnote{Note that the source might also attain 1 from the indicator function if it chooses some other appropriate $\calD\in L_x$, but no such $\calD\in L_x$ will result in challengers choosing $\hat{x}_i > x$ at the signal-truthful equilibrium. Rather, challengers might all agree on $\hat{x}_i < x$.}
Thus, any $\hat{x} > x$ will give strictly lower payoff to the source than $x$, because the indicator will be 0 for any $\hat{x} > x$, while at $x$, it will be 1.
\end{proof}

\section*{Acknowledgments}
The authors thank Pranav Garimidi, Guy Wuollet, and seminar audiences at a16z crypto for helpful comments.

\newpage
\printbibliography
\newpage
\appendix
\section{Proofs Omitted from the Main Text}
\label{app:loc}

\subsection{Proof of \Cref{lem:loc_first}}
\begin{proof}
Assume the contrary, i.e., that there are two $x,y\in\conv(\{p_1,\dots,p_n\})$ such that $x\ne y$ and $\forall i: \dist(p_i, x) = \dist(p_i, y)$.
Rearranging, we obtain
\[
\< x-y, p_i \> = \frac{\|x\|^2 - \|y\|^2}{2} = c
\,,
\]
which is a constant independent of $i$.
Since $x, y$ are points that belong to the convex hull, there exist $\lambda_i, \mu_i \ge 0$ such that $\sum_i \lambda_i = \sum_i \mu_i = 1$ and $x = \sum_i \lambda_i p_i,\; y = \sum_i \mu_i p_i$.
We compute
\[
\< x-y, x \> = \sum_i \lambda_i \<x-y, p_i\> = c \sum_i \lambda_i = c
\,,
\]
and similarly $\<x-y, y\> = c$. Thus, $\|x-y\|^2 = \<x-y, x-y\> = 0$, therefore $x=y$.
This is a contradiction.
\end{proof}

\subsection{Proof of \Cref{lem:loc_sec}}
\begin{proof}
By $x'\in \conv(\{p_1,\dots,p_n\})$, there exist $\lambda_i \ge 0$ such that $\sum_i \lambda_i = 1$ and $x' = \sum_i \lambda_i p_i$.
We calculate
\[
\|p_i-x'\|^2 - \|p_i-x\|^2
=
\|x-x'\|^2 - 2 \< x-p_i, x-x' \>
\,,
\]
and then by weighing and summing the square of all inequalities of the lemma's statement, we obtain that
\[
0
\le
\sum_i \lambda_i \left( \|p_i-x'\|^2 - \|p_i-x\|^2 \right)
=
\|x-x'\|^2 - 2 \left\< x-\sum_i \lambda_i p_i, x-x' \right\>
=
-\|x-x'\|^2
\,,
\]
therefore it has to be that $x'=x$, since the square norm is non-negative.
\end{proof}

\end{document}